\documentclass[a4paper,11pt]{paper}

\usepackage{float}
\usepackage[dvips]{graphics} 
\usepackage[pdftex]{graphicx}
\usepackage{array}
\usepackage{amsmath,amssymb,amsthm,color,amsfonts,graphics,graphicx,verbatim}
\usepackage{algorithm}
\usepackage{algorithmic}
\usepackage{multirow}
\usepackage{appendix}
\usepackage{natbib} 
\usepackage[colorlinks]{hyperref}

\newtheorem{proposition}{Proposition}[section]
\newtheorem{corollary}{Corollary}[section]
\newtheorem{Remark}{Remark}[section]
\theoremstyle{example}
\usepackage{rotating}

\hypersetup{
                bookmarksnumbered,
                pdfstartview={FitH},
                citecolor={blue},
                linkcolor={blue},
                urlcolor={blue},
                pdfpagemode={UseOutlines}
}

\DeclareMathOperator*{\argmax}{\arg\max}
\DeclareMathOperator*{\argmin}{\arg\min}

\def\bfu{{\bf u}}
\def\bfv{{\bf v}}
\def\bfh{{\bf h}}
\def\bfx{{\bf x}}
\def\bfz{{\bf z}}
\def\bfb{{\bf b}}
\def\bfX{{\bf X}}
\def\bfP{{\bf P}}

\def\bfA{{\bf A}}
\def\bfI{{\bf I}}

\def\bfzero{{\bf 0}}
\def\bfone{{\bf 1}}
\def\bfmu{\boldsymbol{\mu}}
\def\bftheta{\boldsymbol{\theta}}
\def\bfSigma{\boldsymbol{\Sigma}}

\title{Cauchy robust principal component analysis with applications to high-deimensional data sets}
\vspace{3in}
\author{Ayisha Fayomi $^1$, Yannis Pantazis $^2$, Michail Tsagris $^3$ \\
and Andrew T.A. Wood $^4$ \\
\\
$^1$ Department of Statistics, King Abdulaziz University, \\
Jeddah, Saudi Arabia \\
\href{mailto:afayomi@kau.edu.sa}{afayomi@kau.edu.sa} \\
$^2$ Institute of Applied \& Computational Mathematics, \\
Foundation for Research \& Technology - Hellas,
Heraklion, Greece \\
\href{mailto:pantazis@iacm.forth.gr}{pantazis@iacm.forth.gr} \\
$^3$ Department of Economics, University of Crete, \\
Gallos Campus, Rethymnon, Greece  \\
\href{mailto:mtsagris@uoc.gr}{mtsagris@uoc.gr} \\
$^4$ Research School of Finance, Actuarial Studies \& Statistics, \\
Australian National University, Canberra, Australia \\
\href{mailto:Andrew.Wood@anu.edu.au}{Andrew.Wood@anu.edu.au}
}

\begin{document}
\maketitle

\begin{abstract}
Principal component analysis (PCA) is a standard dimensionality reduction technique used in various research and applied fields. From an algorithmic point of view,  classical PCA can be formulated in terms of operations on a multivariate Gaussian likelihood. As a consequence of the implied Gaussian formulation, the principal components are not robust to outliers. In this paper, we propose a modified formulation, based on the use of a multivariate Cauchy likelihood instead of the Gaussian likelihood, which has the effect of robustifying the principal components.  We present an algorithm to compute these robustified principal components. We additionally derive the relevant influence function of the first component and examine its theoretical properties. Simulation experiments on high-dimensional datasets demonstrate that the estimated principal components based on the Cauchy likelihood outperform or are on par with existing robust PCA techniques.
\end{abstract}

\section{Introduction}
\label{Intro.3-CPC}
In the analysis of multivariate data, it is frequently desirable to employ statistical methods which are insensitive to the presence of outliers in the sample. To address the problem of outliers, it is important to develop robust statistical procedures. Most statistical procedures include explicit or implicit prior assumptions about the distribution of the observations, but often without taking into account the effect of outliers. The purpose of this paper is to present a novel robust version of PCA which has some attractive features.


Principal components analysis (PCA) is considered to be one of the most important techniques in statistics. However, the classical version of PCA depends on either a covariance or a correlation matrix, both of which are very sensitive to outliers. We develop an alternative method to classical PCA, which is far more robust, by using a multivariate Cauchy likelihood to construct a robust principal components (PC) procedure. It is an adaptation of the classic method of PCA obtained by replacing the Gaussian log-likelihood function by the Cauchy log-likelihood function, in a sense that will be explained in section \ref{Lik.Inter.PCA}.  Although we do not claim that the interpretation of standard PCA in terms of operations on a Gaussian likelihood is new, see Bolton and Krzanowski, this fact does not appear to have been exploited in the development of a robust PCA procedure, as we do in this paper.  An important reason for using the multivariate Cauchy likelihood is that this likelihood has only one maximum point, but the single most important motivation is that it leads to a robust procedure.

In the next section we review briefly some of the techniques employed for estimating parameters and for directing a PCA in ways which are robust against the presence of outliers. We also present robustness preliminaries that include some important techniques which are necessary to assess whether the method used is robust or not. In Section \ref{CPCA} we develop the Cauchy-PCA and theoretically explore its robustness properties. Finally, in Section \ref{Comp.Algo} we present the numerical algorithms for creating Cauchy PCs, and also give the results of a number of very high-dimensional real-data and simulated examples.  Our approach is seen to be competitive with, and often gives superior results to, that of the projection pursuit algorithm of Croux et al. (2007, 2013).  Finally we conclude the paper in Section \ref{concl.}. 

\subsection{Literature review on robust PCA} \label{Lit.Review}
It is well known that PCA is an important technique for high-dimensional data reduction. PCA is based on the sample covariance matrix $\hat{{\bf \Sigma}}$ and it involves searching for a linear combination $y_{j}= {\bf u}^{T}{\bf x}_{j}$ of the ${\bf x}$ components of the vector that maximize the sample variance of the components of $y$. According to \citet{Mardia&Kent&Bibby:1979}, the solution will be given by the equation
\[\hat{\bf \Sigma}={\bf U \Lambda U}^{T},\]
where ${\bf \Lambda}= \hbox{diag}\{\lambda_{1}, \ldots, \lambda_{p}\}$ and its diagonal elements $\lambda_{i}$ are the sample variances, while ${\bf U}$ is an orthogonal matrix, i.e. ${\bf U U}^{T} ={\bf U}^{T}{\bf U}={\bf I}_{p}$, whose columns ${\bf u}_{i}$ are the corresponding eigenvectors which represent the linear combinations.
[[The principal components are efficiently estimated in practice via Singular Value Decomposition (SVD) (cite Lanczos for an efficient algorithm).]]

Classical PCA, unfortunately, is non-robust, since it based on the sample covariance or sample correlation matrix which are very sensitive to outlying observations; see section \ref{NonRob.SPCA}. However, this problem has been handled by two different methods which result in robust versions of PCA by:
 \begin{description}
   \item[i.] replacing the standard covariance or correlation matrix with a robust estimator; or
   \item[ii.] maximising (or minimising) a different objective function to obtain a robust PCA.
 \end{description}
Many different proposes had been developed to carry out robust PCA, such that using projection pursuit PP, $M-$estimators and so on. 

Despite maximum likelihood estimation, perhaps, being considered as the most important statistical inference method, sometimes this approach can lead to improper results when the underlying assumptions are not satisfied, for instance, when data contain outliers or deviate slightly from the supposed model. A generalization of maximum likelihood estimation proposed by \citet{Huber:1964} which is called $M$-estimation, aims to produce a robust statistic by constructing approaches that are resistant to deviations from the underline assumptions. $M$-estimators were also defined for the multivariate case by \citet{Maronna:1976}. 

\citet{Campbell:1980} provided a procedure for robust PCA by examining the estimates of means and covariances which are less affected by outlier observations, and by exploring the observations which have a large effect on the estimates. He replaced the sample covariance sample by an $M-$estimator. \citet{Hubert:2003} introduced a new approach to create robust PCA. It combines the advantages of two methods, the first one is based on replacing the covariance or correlation matrix by its robust estimator, while the second one is based on maximizing the objective function for this robust estimate.

A robust PCA based on the projection pursuit (PP) method was developed by \citet{Li:1985}, using Huber's $M$-estimator of dispersion as the projection index. The objective of PP is to seek projections, of the high-dimensional data set onto low-dimensional subspaces, that optimise a function of "interestingness". The function that should be optimised is called an index or objective function and its choice depends on a feature that the researcher is concerned about. This property gives the PP technique a flexibility to handle many different statistical problems range from clustering to identifying outliers in a multivariate data set. 

\citet{Bolton:1999} characterized the PC's for PP in terms of maximum likelihood under the assumption of normality. PCA can be considered as a special case of PP as well as many other methods of multivariate analysis. \citet{Li:1985} used Huber's $M$-estimator of dispersion as projective index to develop a robust PCA based on the PP approach. The sample median was used as a projective index to develop a robust PCA by \citet{Xie:1993}. In their simulation studies, \citet{Xie:1993} observed a PCA resistant to outliers and deviations from the normal distribution.
\cite{croux2007algorithms,croux2013robust} also suggested a robust PCA using projection pursuit and we will contrast our methodology against their algorithm.  

\section{Preliminaries on standard PCA} \label{NonRob.SPCA}
PCA is an orthogonal linear transformation that projects the data to a new coordinate system according to the variance of each direction. Given a data matrix $\bfX\in\mathbb R^{n\times p}$ with each row correspond to a sample, the first direction $\bfu_1$ that  maximizes the variance is defined through
\begin{equation*}
\bfu_1 = \argmax_{||\bfu||_2=1} ||(\bfX - \bfone_n\bar{\bfx}^T) \bfu||_2^2,
\end{equation*}
where $\bfone_n$ is an $n$-dimensional vector whose elements are all set to 1 while $\bar{\bfx}=\frac{1}{n}\sum_{i=1}^n \bfx_i$ is the empirical mean. 
The process is repeated $k$ times and at each iteration the to-be-estimated principal direction has to be orthogonal to all previously-computed principal directions. Thus, the 
$k$-th direction which has to be orthogonal to the previous ones is defined by
\begin{equation*}
\bfu_k = \argmax_{||\bfu||_2=1} ||(\bfX - \bfone_n\bar{\bfx}^T) \bfu||_2^2 \ \ \text{subject to} \ \  \bfu_k \perp \bfu_j \ \text{with} \ j=1,...,k-1 \ .
\end{equation*}

\subsection{Non-robustness of standard PCA}
We will show that the influence function for the largest eigenvalue of the covariance matrix and the respective eigenvector are unbounded with respect to the norm of an outlier sample. Suppose that $\bfSigma$ is the covariance matrix of a population with distribution function $F$, i.e.,
\begin{equation}\label{Pop.Cov.}
{\bfSigma} = \int_{\mathbb R^p} (\bfx-\bfmu)(\bfx-\bfmu)^{T} dF(\bfx),
\end{equation}
where $\bfmu=\int_{\mathbb R^p} \bfx dF(\bfx)$ corresponds to the mean vector. Assume that the leading eigenvalue of $\bfSigma$ has multiplicity 1, then we denote it by $\lambda$ and the leading eigenvector by $\hat{\bfu}$ (i.e., $\bfu_{1}=\hat{\bfu}$).

Let $T$ be an arbitrary functional, $F$ a distribution and ${\bf z}\in\mathbb R^p$ an arbitrary point in the relevant sample space. The influence function is defined as 
\begin{equation}
IF_T(\bfz;F) = \lim_{\epsilon\to 0+} \frac{T((1-\epsilon) F + \epsilon \Delta_{\bfz}) - T(F)}{\epsilon},
\end{equation}
where $\Delta_{\bfz}$ is a unit point mass located at $\bfz$.

A robust estimator for $T$ means that the influence function is bounded with respect to the norm of the outlier $\bf z$.


\begin{proposition}
The influence function for the leading eigenvector of $\bfSigma$ is given by\footnote{We use ${\bf A}^+$ to denote the Moore-Penrose inverse of a matrix $\bf A$.} 
\begin{equation}
IF_{\hat{\bfu}} (\bfz, F) = - \big( (\bfz-\bfmu)^{T}\hat{\bfu} \big) (\bfSigma-\lambda \bfI_p)^{+} (\bfz-\bfmu).
\end{equation}
Similarly, the IF for the largest eigenvalue of ${\bfSigma}$ is
\begin{equation}
IF_\lambda (\bfz, F) = 
\big( (\bfz-\bfmu)^{T}\hat{\bfu} \big)^2 - \lambda. 
\end{equation}
\end{proposition}

The detailed calculations are presented in Appendix \ref{NonRob:PCA:proof}. The following result shows that outliers with unbounded influence function do exist.

\begin{corollary}
Let $\bfz=\bfmu + \gamma \hat{\bfu} + \eta \bfv$ where $\bfv$ is orthogonal to $\hat{\bfu}$ and does not belong to the null space of $\bfSigma$ and $\gamma,\eta\neq 0$ then
\begin{equation*}
\lim _{\bfz: \, ||\bfz||_2 \rightarrow \infty}||IF_{\hat{\bfu}} (\bfz, F)||_2 = \infty,
\end{equation*}
and similarly for $IF_\lambda (\bfz, F)$.
\end{corollary}

\begin{proof}
Direct substitution of $\bfz$ into the influence function gives:
\begin{equation*}
IF_{\hat{\bfu}} (\bfz, F) = -((\gamma \hat{\bfu} + \eta \bfv)^T \hat{\bfu}) (\bfSigma-\lambda \bfI_p)^{+} (\gamma \hat{\bfu} + \eta \bfv)
= - \gamma \eta (\bfSigma-\lambda \bfI_p)^{+} \bfv.
\end{equation*}
Since $\bfv$ does not belong to the null space of $\bfSigma$, it holds that $(\bfSigma-\lambda \bfI_p)^{+} \bfv \neq \bfzero$ thus $||(\bfSigma-\lambda \bfI_p)^{+} \bfv||_2=c\neq0$. Hence,
\begin{equation*}
||IF_{\hat{\bfu}} (\bfz, F)||_2 = |\gamma| |\eta| c.
\end{equation*}
Given that $||\bfz||_2^2 = \gamma^2+\eta^2+||\bfmu||_2^2+\gamma \bfmu^T\hat{\bfu}+\eta \bfmu^T\bfv$, either sending $|\gamma|\to\infty$ or $|\eta|\to\infty$ completes the proof.

Similarly,
\begin{equation*}
IF_\lambda (\bfz, F) = \gamma^2-\lambda \rightarrow \infty, 
\end{equation*}
as $|\gamma|\to\infty$.
\end{proof}

\subsection{Generalizations of standard PCA}
\label{Lik.Inter.PCA}
Standard PCA can be viewed as a special case of a more general optimization problem. We present two such generalization: the first one leads to projection pursuit algorithms while the second leads to a maximum likelihood formulation. Let $\bfu$ be a unit vector and define the projection values
\begin{equation*}
c_{i}(\bfu) = \bfx^{T}_{i} \bfu, {\hspace{3mm}} i=1, \ldots, n,
\end{equation*}
and a function $\Phi:\mathbb R^n \to \mathbb R$ acting on the projected values. The first generalization of PCA is defined as the maximization of $\Phi$:
\begin{equation*}
\bfu_1 = \argmax_{||\bfu||_2=1} \Phi(c_1(\bfu),...,c_n(\bfu)) \ .
\end{equation*}
As in the standard PCA, the following principal directions are obtained after removing the contribution of the current principal component from the data. When $\Phi$ is the sample variance then we recover the standard PCA.

The second generalization interprets the computation of the principal component as a maximum likelihood estimation problem. By letting, 
\begin{equation}\label{GausLogLik}
    l_{G}(\mu, \sigma^{2}| c_{1},\ldots, c_{n})= -\frac{n}{2} \log {\sigma}^{2} - \frac{n}{2{\sigma}^{2}}\sum_{i=1}^{n}(c_{i}-\mu)^{2}.
\end{equation}
be the Gaussian log-likelihood, the first principal direction can be obtained by solving the minimax problem:
\begin{equation*}
\min_{||\bfu||_2=1}\max_{\mu, \sigma^2} \ l_{G}(\mu, \sigma^{2}| c_{1}(\bfu),\ldots, c_{n}(\bfu)).
\end{equation*}
Indeed, the inner maximization can be solved analytically which leads to the optimal solution
\begin{equation*}
\hat{\mu}(\bfu) = \frac{1}{n} \sum_{i=1}^n c_i(\bfu) =: \bar{c}({\bf u})
\end{equation*}
and
\begin{equation*}
{\hat{\sigma}}^{2}({\bf u}) = \frac{1}{n} \sum_{i=1}^{n} (c_{i}({\bf u})- \bar{c}({\bf u}))^{2}.
\end{equation*}
Unsurprisingly, the optimal values are the sample mean and the sample variance. Using the above formulas it is straightforward to show that
\begin{eqnarray}
  \argmin_{||\bfu||_2=1} \ l_{G}\big(\hat{\mu}(\bfu), {\hat{\sigma}}^{2}({\bfu})| c_{1}({\bfu}), \ldots, c_{n}({\bfu})\big)
  = \argmax_{||\bfu||_2=1} \ \hat{\sigma}^{2}({\bfu}) \ .
\end{eqnarray}
Variations of PCA can be derived by changing the likelihood function and in the next section we analyze the case of Cauchy distribution.

\section{Cauchy PCA}
\label{CPCA}
The Cauchy log-likelihood function is given by
\begin{equation}\label{Cau.LogLik}
    l_{C}({\mu},{\sigma}| {c}_{1}({\bfu}), \ldots, {c}_{n}({\bfu}))=  n \log{\frac{\sigma}{\pi}} - \sum_{i=1}^{n} \log \left\{{\sigma}^{2}+ (c_{i}(\bfu)-{\mu})^{2}\right\}.
\end{equation}
where $\mu$ and $\sigma$ are the two parameters of the Cauchy distribution. The first Cauchy principal direction is also obtained by solving the minimax optimization problem:
\begin{equation}\label{cauchy:minimax}
\min_{||\bfu||_2=1}\max_{\mu,\sigma} \ l_{C}(\mu, \sigma^{2}| c_{1}(\bfu),\ldots, c_{n}(\bfu)).
\end{equation}
In contrast to the Gaussian case, the inner maximization cannot be performed analytically. Therefore an iterative approach needs to be utilized. Here, we apply the Newton-Raphson method with initial values the median and half the interquartile range for the location and scale parameters, respectively. According to \citet{Copas:1975}, although the mean of the Cauchy distribution does not exist and it has infinite variance, the Cauchy log-likelihood function $l_{C}(\mu, \sigma)$ has a unique maximum likelihood estimate, $(\hat{\mu},\hat{\sigma})$. 

Fixing $\mu$ and $\sigma$, the outer minimization is also non-analytic and a fixed point iteration is applied to calculate $\bfu$. The iteration is given by
\begin{equation}
\hat{\bfu} = \frac{\hat{\bfu}_{un}}{||\hat{\bfu}_{un}||_2},
\label{cauchy:norm:eq}
\end{equation}
where $\hat{\bfu}_{un}$ is the unnormalized direction which is obtained from the differentiation of the Lagrangian function with respect to $\bfu$ and it is given by
\begin{eqnarray} \label{parallel}
\hat{\bfu}_{un} = \sum_{i=1}^{n}\frac{({\bfx}_i^T{\hat{\bfu}}-\hat{\mu}){\bfx}_i} {\hat{\sigma}^2 + \left({\bfx}_i^T{\hat{\bfu}}-\hat{\mu} \right)^2} \ .
\label{fixed:point:eq}
\end{eqnarray}

Once the first principal direction has been computed, its contribution from the dataset $\bfX$ is removed and the same procedure to estimate the next principal direction is repeated. This iterative process is repeated $k$ times. The removal of the contribution makes the principal directions orthogonal to each other.
We summarize the estimation of $k$ Cauchy principal components in the following pseudo-code (Algorithm \ref{1CPC:Algo.}).

\begin{algorithm}[H]
\caption{Cauchy PCA}
\label{1CPC:Algo.}
\begin{algorithmic}
\FOR{$j=1,...,k$}
\STATE $\bullet$ Initialize ${\hat{\bfu}_{un}}$ and normalize
$\hat{\bfu}= \hat{\bfu}_{un} / ||\hat{\bfu}_{un}||_2$
\WHILE{not converged}
\STATE $\bullet$ Fix $\hat{\bfu}$ and set
$$ c_i(\hat{\bfu}) = \bfx_i^T\hat{\bfu}, \ \ i=1,...,n.$$
\STATE $\bullet$ Via Newton-Raphson algorithm find
$$(\hat{\mu},\hat{\sigma})=\argmax_{\mu, \sigma} \ l_C(\mu, \sigma; c_1(\hat{\bfu}), \ldots, c_n(\hat{\bfu})).$$ 
\STATE $\bullet$ Fix $(\hat{\mu}, \hat{\sigma})$ and using fixed point iteration (i.e., (\ref{fixed:point:eq}) \& (\ref{cauchy:norm:eq})) find
$$\hat{\bfu} = \argmin_{\bfu} \ l_C(\hat{\mu}, \hat{\sigma}| c_1(\bfu), \ldots, c_n(\bfu)) - \lambda (||\bfu||_2^2-1)$$
\ENDWHILE
\STATE $\bullet$ Set the $j$-th Cauchy principal direction
$$\bfu_{j} = \hat{\bfu}.$$
\STATE $\bullet$ Remove the contribution from the dataset
\begin{eqnarray*}
\bfX = \bfX (\bfI_p - \bfu_{j}\bfu^T_{j}),
\end{eqnarray*}
\ENDFOR
\end{algorithmic}
\end{algorithm}

\subsection{Robustness of the Leading Cauchy Principal Direction}
Let $\bftheta = \left(\mu,\sigma\right)^T$ be the parameter vector of the Cauchy distribution and consider the infinite-sample normalized Cauchy log-likelihood function
\begin{equation}
l(\bfu|\bftheta) = \int_{\bfx\in\mathbb R^p} g(c(\bfu),\bftheta)\, dF(\bfx),
\end{equation}
where $g(c,\bftheta) = \log(\sigma/\pi) - \log( \sigma^2+(c-\mu)^2)$ and $c(\bfu)=\bfx^T\bfu$. We will estimate the influence function for the leading Cauchy principal direction
\begin{equation}
\hat{\bfu} = \argmin_{||\bfu||_2=1} \ l(\bfu|\bftheta_F(\bfu)),
\end{equation}
where $\bftheta_F(\bfu)=\argmax_{\bftheta} l(\bfx^T\bfu|\bftheta)$ is the optimal Cauchy parameters for a given direction $\bfu$.

Since $\hat{\bfu}$ is restricted to be a unit vector, the standard condition for the minimum, i.e., $\left.\frac{\partial}{\partial\bfu} l(\bfu|\bftheta_F(\bfu))\right\vert_{\bfu=\hat{\bfu}} = \bfzero$ is not valid. The proper condition is defined by
\begin{equation}
\bfP_{\hat{\bfu}} \left.\frac{\partial}{\partial\bfu} l(\bfu|\bftheta_F(\bfu))\right\vert_{\bfu=\hat{\bfu}} = \bfzero ,
\end{equation}
where $\bfP_{\bfu}$ is the projection matrix given by $\bfP_{\bfu}=\bfI_p-\bfu\bfu^T$.

\begin{Remark}
An equivalent condition is to satisfy  $\bfh^T \left.\frac{\partial}{\partial\bfu} l(\bfu|\bftheta_F(\bfu))\right\vert_{\bfu=\hat{\bfu}} = \bfzero$ for all $\bfh$ such that $\bfh^T\hat{\bfu}=0$ and $||\bfh||_2=1$. Both derived conditions are essentially a consequence of the Lagrangian formulation of the constraint optimization problem. Indeed, the Lagrange condition implies that at the minimum the direction of the objective function's derivative should be parallel to the direction of the constraint's derivative which translates to $\left.\frac{\partial}{\partial\bfu} l(\bfu|\bftheta_F(\bfu))\right\vert_{\bfu=\hat{\bfu}} = \lambda \hat{\bfu}$ where $\lambda\neq 0$ is the Lagrange multiplier.
\end{Remark}

Let $\bar{g}(\bfx;\bfu) = \left.g(\bfx^T\bfu|\theta)\right\vert_{\theta=\theta_F(\bfu)}$ be the likelihood function computed at $\theta=\theta_F(\bfu)$ and let denote its partial derivatives as 
\[
\bar{g}_c(\bfx;\bfu) = \left.\frac{\partial}{\partial c} g(\bfx^T\bfu|\theta)\right\vert_{\theta=\theta_F(\bfu)}
\]
and  
\[
\bar{g}_{\bftheta}(\bfx;\bfu) = \left.\frac{\partial}{\partial \bftheta} g(\bfx^T\bfu|\theta)\right\vert_{\theta=\theta_F(\bfu)}. 
\]
Similarly, $\bar{g}_{cc}$, $\bar{g}_{c\theta}$ and $\bar{g}_{\theta\theta}$ denote the second order derivatives.
The following proposition establishes the expression for the influence function of the leading Cauchy principal direction, $\hat{\bfu}$. 

\begin{proposition}\label{influence:func:cauchy:pca}
Under the assumption of ${\bfI}_F(\hat{\bfu})$ and $\bfA$ being invertible matrices, the influence function of $\hat{\bfu}$ is
\begin{equation}
IF_{\hat{\bfu}} (\bfz, F) = \bfA^{-1} \bfb ,
\end{equation}
where
$$
\begin{aligned}
\bfA &= \bfI_p \int_{\mathbb R^p} \bar{g}_{c\bftheta}(\bfx;\hat{\bfu})  \bfx^T\hat{\bfu} dF(\bfx) 
- \bfP_{\hat{\bfu}} \int_{\mathbb R^p} \bar{g}_{cc}(\bfx;\hat{\bfu}) \bfx^T\bfx dF(\bfx) \bfP_{\hat{\bfu}} \\
&+ \bfP_{\hat{\bfu}} \int_{\mathbb R^p}  \bfx \bar{g}_{c\bftheta}(\bfx;\hat{\bfu}) dF(\bfx)  \, {\bfI}_F(\hat{\bfu})^{-1} \, 
\int_{\mathbb R^p}  \bar{g}_{\bftheta c}(\bfx;\hat{\bfu}) \bfx^T dF(\bfx) \bfP_{\hat{\bfu}}
\end{aligned}
$$
and
$$
\bfb = \bfb(z) = \bar{g}_c(\bfz, \hat{\bfu}) \bfz + \int_{\mathbb R^p} \bfx \bar{g}_{c\bftheta}(\bfx;\hat{\bfu}) dF(\bfx) \, 
{\bfI}_F(\hat{\bfu})^{-1} \,  \bar{g}_{\bftheta}(\bfz;\hat{\bfu}),
$$
while
$$
{\bfI}_F(\hat{\bfu}) = \int_{\mathbb R^p} \bar{g}_{\bftheta\bftheta}(\bfx;\hat{\bfu}) dF(\bfx)
$$
is the expected Fisher information matrix under $F$ for the parameters of the Cauchy distribution computed at $\hat{\bfu}$.
\end{proposition}
\begin{proof}
The proof consists of several straightforward series expansions and implicit function calculations. The complete proof is given in Appendix \ref{robust:cauchy:proof}.
\end{proof}

The following boundedness result for the influence function states the conditions under which Cauchy PCA is robust. 

\begin{corollary} \label{boundness} Let the assumptions of the proposition hold.
If $\bfz\not\perp\hat{\bfu}$ or if $\bfz\perp\hat{\bfu}=0$ but $\mu_F(\hat{\bfu})=0$ then the influence function for $\hat{\bfu}$ is bounded.
\end{corollary}

\begin{proof}
First, observe that matrix $\bfA$ does not depend on $\bfz$. It is only $\bfb$ that depends on $\bfz$ and our goal is to prove that $\bfb$ is bounded with respect to $\bfz$. Second, we have to compute the partial derivatives $\bar{g}_c(\bfz; \hat{\bfu})$ and $\bar{g}_{\bftheta}(\bfz; \hat{\bfu})$. Straightforward calculations lead to
$$
\bar{g}_c(\bfz; \hat{\bfu}) = - \frac{2(\bfz^T\hat{\bfu}-\mu_F(\hat{\bfu}))}{\sigma_F^2(\hat{\bfu})+(\bfz^T\hat{\bfu}-\mu_F(\hat{\bfu}))^2}
$$
$$
\bar{g}_\mu(\bfz; \hat{\bfu}) = \frac{2(\bfz^T\hat{\bfu}-\mu_F(\hat{\bfu}))}{\sigma_F^2(\hat{\bfu})+(\bfz^T\hat{\bfu}-\mu_F(\hat{\bfu}))^2}
$$
and
$$
\bar{g}_\sigma(\bfz; \hat{\bfu}) = \frac{1}{\sigma_F(\hat{\bfu})} - \frac{2\sigma_F(\hat{\bfu})}{\sigma_F^2(\hat{\bfu})+(\bfz^T\hat{\bfu}-\mu_F(\hat{\bfu}))^2}.
$$

Let us now define an arbitrary scaling of the outlier $\bfz\rightarrow\alpha\bfz$ and prove boundedness of $\bfb$ as we send $\alpha\to\infty$. We consider the first case where $\bfz\not\perp\hat{\bfu}$. It holds that $\lim_{\alpha\to\infty} \bar{g}_c(\alpha\bfz; \hat{\bfu})\alpha\bfz = -(\bfz^T\hat{\bfu})^{-1}\bfz$,  $\lim_{\alpha\to\infty} \bar{g}_\mu(\alpha\bfz; \hat{\bfu}) = 0$ and $\lim_{\alpha\to\infty} \bar{g}_\sigma(\alpha\bfz; \hat{\bfu}) = \frac{1}{\sigma_F(\hat{\bfu})}$ therefore $\bfb$ is bounded with respect to $\alpha$.

For the second case, we have 
\[
\lim_{\alpha\to\infty} \bar{g}_c(\alpha\bfz; \hat{\bfu})\alpha\bfz = \lim_{\alpha\to\infty} \frac{2\mu_F(\hat{\bfu})}{\sigma_F^2(\hat{\bfu})+\mu_F(\hat{\bfu})^2} \alpha\bfz = 0,
\]
\[
\lim_{\alpha\to\infty} \bar{g}_\mu(\alpha\bfz; \hat{\bfu}) = \frac{2\mu_F(\hat{\bfu})}{\sigma_F^2(\hat{\bfu})+\mu_F(\hat{\bfu})^2} = 0
\]
and
\[
\lim_{\alpha\to\infty} \bar{g}_\sigma(\alpha\bfz; \hat{\bfu}) = 
\frac{1}{\sigma_F(\hat{\bfu})} - \frac{2\sigma_F(\hat{\bfu})}{\sigma_F^2(\hat{\bfu})+\mu_F(\hat{\bfu})^2}
= -\frac{1}{\sigma_F(\hat{\bfu})}
\]
since $\mu_F(\hat{\bfu})=0$
by assumption. Thus $\bfb$ is bounded with respect to $\alpha$ for the second case, too.
\end{proof}

The only case not covered by the corollary is when $\bfz^T\hat{\bfu}=0$ and $\mu(\hat{\bfu})\neq 0$. Our experiments presented in the following section show that outliers that are orthogonal to the Cauchy principal direction do sometimes influence the estimation of the Cauchy principal direction yet not significantly.

\subsection{Several Cauchy principal components}
 We briefly mention possibilities for estimating several Cauchy principal components.  There are two obvious approaches: one approach, the sequential approach, is to repeat the algorithm described above on the subspace orthogonal to $\hat{\bfu}=\hat{\bfu}_1$ to obtain $\hat{\bfu}_2$, the second Cauchy principal component, where $\hat{\bfu}_1$ is the first Cauchy principal component; then  repeat the procedure on the subspace orthogonal to $\hat{\bfu}_1$ and $\hat{\bfu}_2$ to obtain $\hat{\bfu}_3$; and so on.  A second approach, the simultaneous approach, is to decide in advance how many principal components we wish to determine, $p$ say, and then use a $p$-dimensional multivariate Cauchy likelihood, which has $p+ p(p+1)/2$ free parameters, to obtain $\hat{\bfu}_1, \ldots , \hat{\bfu}_p$.  It turns out that these two approaches lead to equivalent results in classical (Gaussian) PCA but when a Cauchy likelihood is used the two approaches produce different sets of principal components.  Our current thinking is this: the sequential approach is easier to implement (essentially the same software can be used at each step) and it is faster.  However, the simultaneous approach could potentially be preferable if we know in advance how many principal components we wish to estimate. Further investigation is required.


\section{Numerical Results}
\label{Comp.Algo}

\subsection{Simulation studies}
In this section we will empirically validate our proposed methodology, via simulation studies. We searched for R packages that offer robust PCA in the $n<<p$ case and came up with \textit{FastHCS} \citep{fasthcs2018}, \textit{rrcovHD} \citep{rrcovhd2016}, \textit{rpca} \citep{rpca2017} and \textit{pcaPP} \citep{pcapp2018}. Out of them, \textit{pcaPP} (Projection Pursuit PCA) is the only one which does not require hyper-parameter tuning, e.g. selection of the LASSO penalty $\lambda$ or choice of the percentage of observations used to estimate a robust covariance matrix. 

\subsubsection{Setup of the simulations}
Initially,  we created a $p \times p$ (orthonormal) basis $\bf B$ by using QR decomposition on some randomly generated data. We then generated eigenvalues $\lambda_i \sim Exp(0.4)$, where $i=1,\ldots,p$ and hence we obtained the covariance matrix $\pmb{\Sigma} = {\bf B}\pmb{\Lambda}{\bf B}^T$, where $\pmb{\Lambda} =\text{diag}(\lambda_i)$. The first column of $\bf B$ served as the first ``clean'' eigenvector, and was the benchmark in our comparative evaluations. Following this step, we simulated $n$ random vectors ${\bf X} \sim N_p\left({\bf 0}, \pmb{\Sigma} \right)$ and in order to check the robustness of the results to the center of the data, all observations were shifted right by adding $50$ everywhere. A number of outliers equal to 2$\%$ of the sample size were introduced. These outliers were $\bar{\bf x}+e^{\kappa}{\bf z} \in {\mathbb{R}}^{p}$, where $\bar{\bf x}$ is the sample mean vector, ${\bf z}$ are unit vector(s) and $e^{\kappa}$ a real number denoting their norm, where $\kappa$ varied from $3$ up to $8$ increasing with a step size equal to $1$ and the angle between the outliers ${\bf z}$ and the first ``clean'' eigenvector spanned from $0^{\circ}$ up to $90^{\circ}$. In all cases, we subtracted the spatial median or the column-wise median\footnote{The results are pretty similar for either type of median and we here show the results of he column-wise median.} and scaled them by the mean absolute deviation.

At each case, we computed the first Cauchy-PCA eigenvector and the first PP-PCA eigenvector. The performance metric is the angle (in degrees) between the first robust (based on Cauchy or PP-PCA) eigenvector and the first "clean" eigenvector computed using the classical PCA. All experiments were repeated $100$ times and the results were averaged. 

\subsubsection{Comparative results}

Tables \ref{tab100_500}-\ref{tab500_1000} present the performance of the first Cauchy-PCA eigenvector and of the first PP-PCA eigenvector for a variety of norms of the outlier, with different angles ($\phi$) between the outlier and the leading true eigenvector, for the $n<p$ case. 

The case of $n<p$ was selected as statistical inference in this case is more challenging than the $p<n$  case\footnote{In this paper we focus on high-dimensional simulations and real-date examples ($p>n)$  but in results not presented in the paper we found that Cauchy PCA is also very competitive and performs strongly in low dimensional settings ($p<n$).}. Additionally, this case is also ordinarily met in the field of bioinformatics were the -omics data count tens of thousands of variables (genes, single nucleotide polymorphisms, etc.) but only tens or at most hundreds of observations. 

As observed in Tables \ref{tab100_500}-\ref{tab500_1000}, the average angular difference between the Cauchy and the PP PCA ranges from $20^{\circ}$ up to more than $50^{\circ}$, which is evidently quite substantial, providing evidence that Cauchy PCA has performed in a superior manner to the projection pursuit method of Croux et al. (2007, 2013). In particular, the tables demonstrate that Cauchy PCA is less error prone than its competitor but, as is seen in Table \ref{tab500_1000}, the error decreases for both methods with increasing sample size. Further, the mean angular difference between the two methods increases as the angle $\phi$ increases. For instance, in Table \ref{tab100_500}, when $k=8$ and $\phi=0^{\circ}$ the difference between the two methods is $20^{\circ}$, whereas when $\phi=90^{\circ}$ the difference increases to $48^{\circ}$. Further, the error is not highly affected by the angle $\phi$, or the norm of the outliers. It can be seen  that in Table \ref{tab100_1000} and Table \ref{tab500_1000} in the special case of $\phi=90^{\circ}$, the error increases for the Cauchy PCA by $2^{\circ}-3^{\circ}$, thus corroborating the result of Corollary \ref{boundness}. However, this effect, as in Table \ref{tab100_500}, is rather small, though noticeable. 

\begin{table}
\caption{Mean angular difference between the robust eigenvectors computed in the contaminated data and the sample eigenvector computed in the clean data when $n=100$ and $p=500$. The norm of the outliers is $e^{k}$ and their angle with the true clean eigenvector is denoted by $\phi$.}
\label{tab100_500}
\begin{tabular}{ll|rrrrrrr}
\hline
Angle  &  Method  & k=-Inf & k=3 & k=4 & k=5 & k=6 & k=7 & k=8 \\ \hline
$\phi=0^{\circ}$  & Cauchy & 31.17 & 29.79 & 29.54 & 28.83 & 28.86 & 29.24 & 28.78 \\
                  & PP     & 82.45 & 49.91 & 48.84 & 48.22 & 49.08 & 49.61 & 48.14 \\ \hline
$\phi=30^{\circ}$ & Cauchy & 31.44 & 29.24 & 29.13 & 28.60 & 28.89 & 29.34 & 29.65 \\
                  & PP     & 82.45 & 65.28 & 65.34 & 63.42 & 62.96 & 66.63 & 65.43 \\ \hline
$\phi=60^{\circ}$ & Cauchy & 31.49 & 29.86 & 29.07 & 29.04 & 29.55 & 29.70 & 29.09 \\ 
                  & PP     & 82.11 & 81.11 & 82.55 & 82.63 & 82.12 & 82.49 & 82.03 \\ \hline
$\phi=90^{\circ}$ & Cauchy & 32.32 & 31.67 & 33.00 & 33.13 & 32.86 & 33.19 & 33.06 \\ 
                  & PP     & 82.38 & 82.06 & 81.69 & 82.12 & 81.73 & 81.74 & 81.88 \\ \hline 
\end{tabular}
\end{table}

\begin{table}
\caption{Mean angular difference between the robust eigenvectors computed in the contaminated data and the sample eigenvector computed in the clean data when $n=100$ and $p=1000$. The norm of the outliers is $e^{k}$ and their angle with the true clean eigenvector is denoted by $\phi$.}
\label{tab100_1000}
\begin{tabular}{ll|rrrrrrr}
\hline
Angle  &  Method  & k=-Inf & k=3 & k=4 & k=5 & k=6 & k=7 & k=8 \\ \hline
$\phi=0^{\circ}$  & Cauchy & 36.53 & 33.12 & 33.60 & 33.69 & 32.62 & 32.51 & 33.16 \\
                  & PP     & 83.06 & 80.36 & 80.17 & 81.87 & 80.50 & 80.76 & 80.16 \\ \hline
$\phi=30^{\circ}$ & Cauchy & 36.55 & 34.72 & 33.91 & 33.09 & 33.11 & 33.16 & 32.79 \\ 
                  & PP     & 83.07 & 82.36 & 82.76 & 82.65 & 83.07 & 82.93 & 83.12 \\ \hline
$\phi=60^{\circ}$ & Cauchy & 36.42 & 34.46 & 33.96 & 33.61 & 34.41 & 33.07 & 33.47 \\
                  & PP     & 83.78 & 82.86 & 82.71 & 84.05 & 83.46 & 82.71 & 82.78 \\ \hline
$\phi=90^{\circ}$ & Cauchy & 36.50 & 36.12 & 36.81 & 37.18 & 39.34 & 39.11 & 38.51 \\
                  & PP     & 83.63 & 83.73 & 83.69 & 83.65 & 84.03 & 83.66 & 83.00 \\ \hline
\end{tabular}
\end{table}

\begin{table}
\caption{Mean angular difference between the robust eigenvectors computed in the contaminated data and the sample eigenvector computed in the clean data when $n=500$ and $p=1000$. The norm of the outliers is $e^{k}$ and their angle with the true clean eigenvector is denoted by $\phi$.}
\label{tab500_1000}
\begin{tabular}{ll|rrrrrrr}
\hline
Angle  &  Method  & k=-Inf & k=3 & k=4 & k=5 & k=6 & k=7 & k=8 \\ \hline
$\phi=0^{\circ}$  & Cauchy & 19.95 & 18.60 & 18.46 & 18.35 & 18.24 & 18.20 & 17.93 \\
                  & PP     & 68.76 & 26.08 & 24.93 & 24.91 & 24.83 & 24.73 & 24.72 \\ \hline 
$\phi=30^{\circ}$ & Cauchy & 19.43 & 18.30 & 18.39 & 18.22 & 18.16 & 18.01 & 18.13 \\ 
                  & PP     & 68.98 & 39.72 & 38.88 & 38.44 & 38.20 & 38.15 & 38.14 \\ \hline
$\phi=60^{\circ}$ & Cauchy & 19.76 & 18.60 & 18.12 & 18.20 & 18.40 & 18.19 & 18.01 \\
                  & PP     & 69.10 & 64.10 & 63.12 & 62.89 & 62.91 & 62.82 & 62.77 \\ \hline
$\phi=90^{\circ}$ & Cauchy & 19.49 & 19.84 & 20.16 & 21.87 & 22.41 & 22.87 & 22.84 \\
                  & PP     & 68.99 & 68.62 & 68.59 & 68.70 & 68.45 & 68.73 & 68.43 \\ \hline
\end{tabular}
\end{table}

\subsection{High dimensional real datasets}
Two real gene expression datasets, GSE13159 and GSE31161\footnote{From a biological standpoint, the data have already been uniformly pre-processed, curated and automatically annotated.}, downloaded from the \href{dataome.mensxmachina.org}{Biodataome} platform \citep{lakiotaki2018}, were used in the experiments. The dimensions of the datasets were equal to $2,096 \times 54,630$ and $1035 \times 54,675$, respectively. We randomly selected $5,000$ variables and computed the outliers using the high dimensional Minimum Covariance Determinant (MCD) of \cite{ro2015}. In accordance with the simulations studies, we removed the $2\%$ of the most extreme outliers detected by MCD and computed the first classical PC (benchmark eigenvector), the first Cauchy-PCA eigenvector and the first PP-PCA eigenvector of the "clean" data. We then added those outliers and increased their norm by $e^k$, where $k=(0, 3, 4, \ldots, 8)$ and computed computed the first Cauchy-PCA eigenvector and the first PP-PCA eigenvector. In all cases, we subtracted the spatial median or the column-wise median and scaled them by the mean absolute deviation. The performance metric is the angle (in degrees) between the first robust (based on Cauchy or PP-PCA) eigenvector and the first true ``clean" eigenvectors and the time required by each method. This procedure was repeated  $200$ times and the average results are graphically displayed in Figures \ref{gse}(a)-(d). 

Broadly speaking the effect of the PP PCA does not seem to have been affected substantially by the centering method, i.e. subtraction of the spatial or the column-wise median. On the contrary, the Cauchy PCA is affected by the type of median employed to this end. Centering with the spatial median yields high error levels for all norms of the outliers, for both datasets, whereas centering with the column-wise median produces much lower error levels. On average, the difference in the error between Cauchy PCA and PP PCA is about $30^{\circ}$ for the GSE31159 dataset (Figure \ref{gse}(a)) and about $14^{\circ}$ for the GSE3161 dataset (Figure \ref{gse}(b)). However, the error of the Cauchy PCA increases and the stabilizes in the GSE31159 dataset whereas the error of the PP PCA is stable regardless of the norm of the outliers. A different conclusion is extracted in the GSE31161 where the error of either method decreases as the norm of the outliers increases, until it reaches a plateau. 

With regards to computational efficiency, the PP PCA is not affected by either centering method, whereas Cauchy PCA seems to be affected in the GSE31159 dataset but not in the GSE31161 dataset as seen in Figures \ref{gse}(c) and \ref{gse}(d). Cauchy PCA centered with the column-wise median is, on average, 5 times faster than PP PCA.

\begin{figure}[!ht]
\centering
\begin{tabular}{cc}
GSE31159  &  GSE31161    \\
\includegraphics[scale = 0.4, trim = 30 0 0 0]{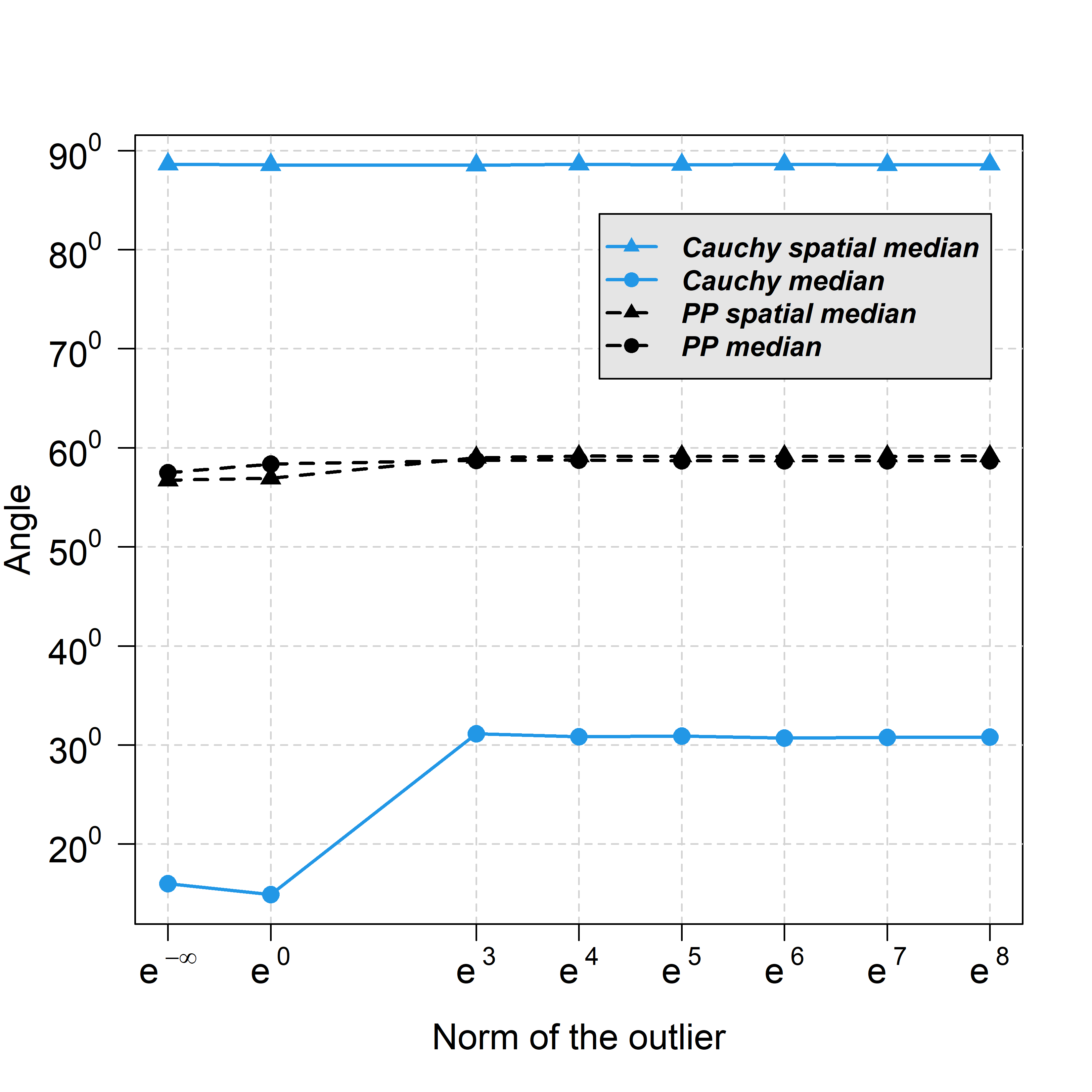} &
\includegraphics[scale = 0.4, trim = 30 0 0 0]{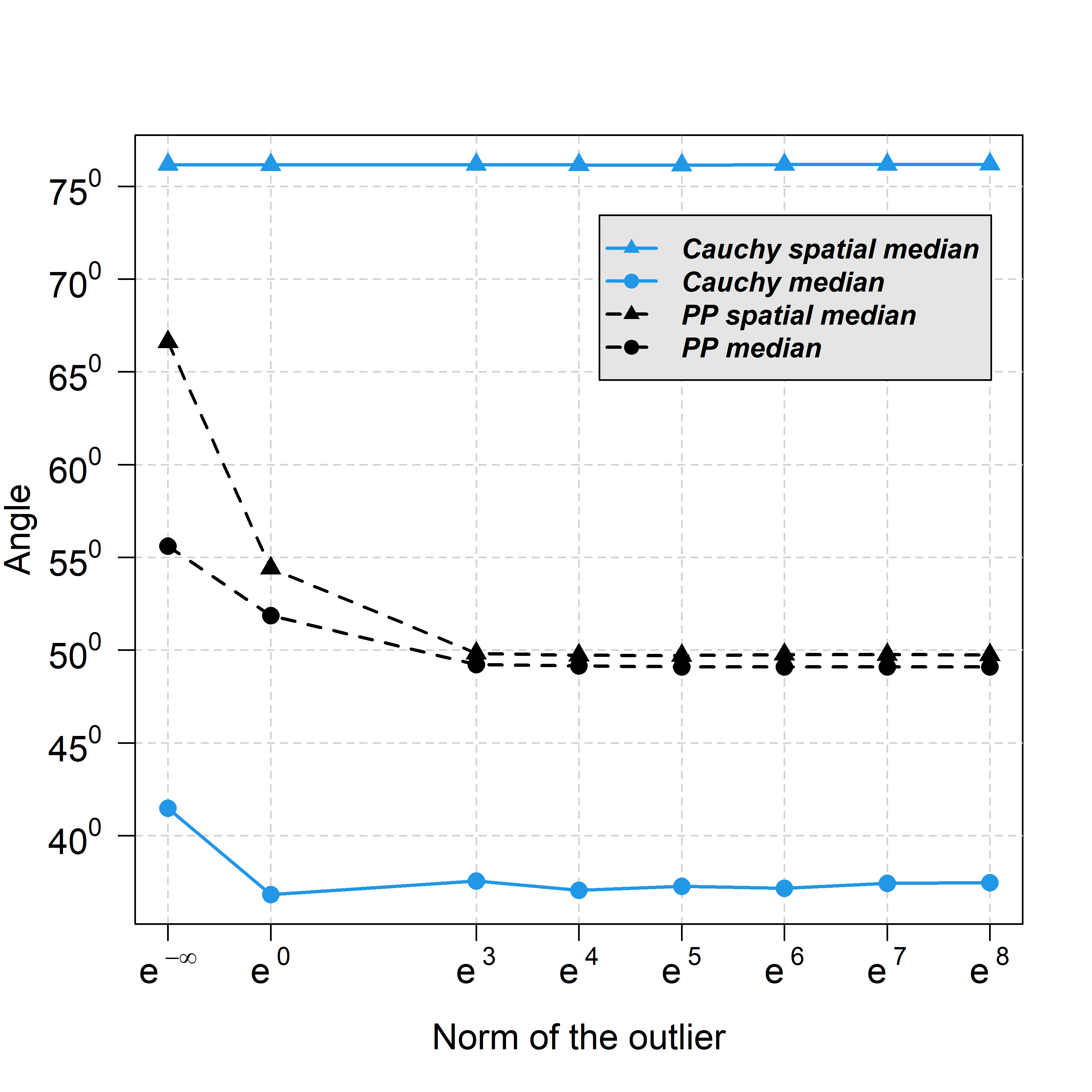} \\
(a)    &  (b)      \\ 
\includegraphics[scale = 0.4, trim = 30 0 0 0]{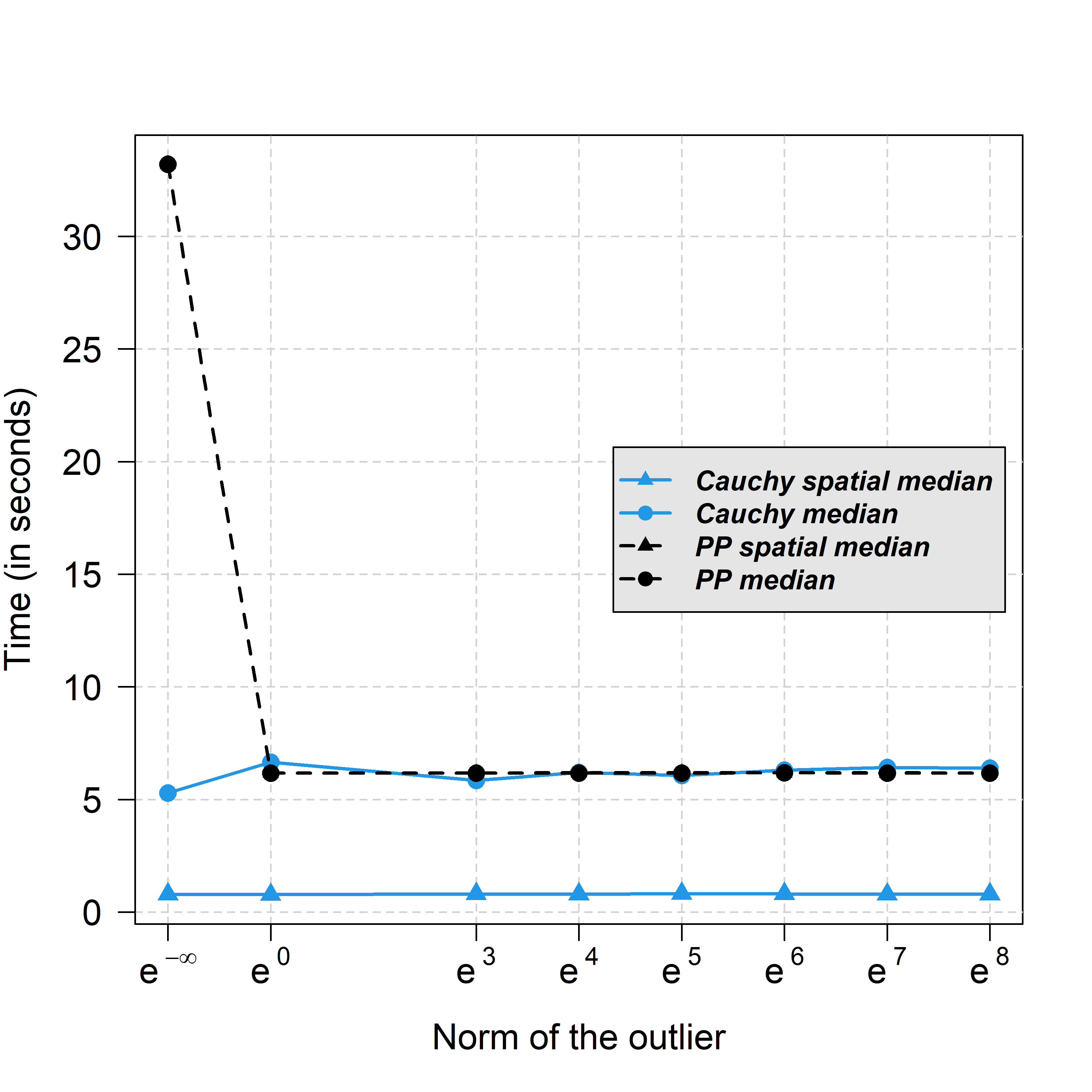} &
\includegraphics[scale = 0.4, trim = 30 0 0 0]{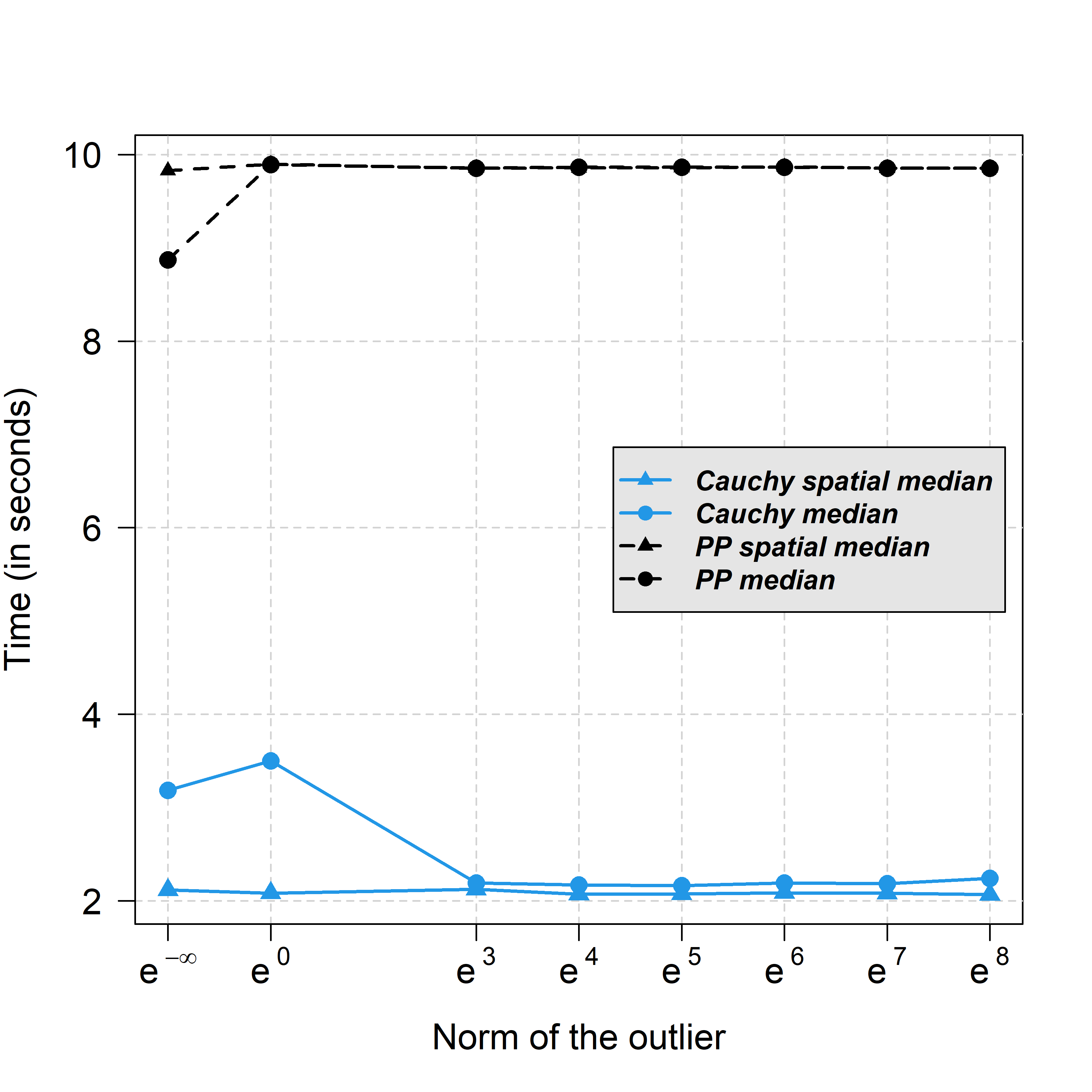} \\
(c)    &  (d)  
\end{tabular}
\caption{The first row presents the angle between the first Cauchy PC of the "contaminated" data and the 1st leading eigenvector of the "clean" data and the angle between the first Projection Pursuit PC of the "contaminated" data and the 1st leading eigenvector of the "clean" data for increasing norms of the outliers. The second row contains the time in seconds.}
\label{gse}
\end{figure}

\section{Conclusion}\label{concl.}
The starting point for this paper is the observation that classical PCA can be formulated purely in terms of operations on a Gaussian likelihood. Although this observation is not new, the specifics of this formulation of classical PCA do not appear to be as widely known as might be expected. The novel idea underlying this paper is to formulate a version of PCA in which a Cauchy likelihood is used instead of a Gaussian likelihood, leading to what we call Cauchy PCA. Study of the resulting influence functions shows that Cauchy PCA has very good robustness properties. Moreover, we have provided an implementation of Cauchy PCA which runs quickly and reliably. Numerous simulation and real-data  examples, mainly in high-dimensional settings, show that Cauchy PCA typically out-performs alternative robust versions of PCA whose implementation is in the public domain. 

\clearpage
\section*{Appendix}
\setcounter{section}{0}
\renewcommand{\thesubsection}{A\arabic{subsection}}

\subsection{Proof of Proposition 2.1}\label{NonRob:PCA:proof}
\begin{proof}
The perturbed distribution $(1-\epsilon)F(\bfx) + \epsilon\Delta_\bfz(\bfx)$ has perturbed mean value
\begin{equation*}
\bfmu_\epsilon = \bfmu + \epsilon (\bfz-\bfmu)
\end{equation*}
and perturbed covariance matrix
\begin{equation*}
\bfSigma_\epsilon = \bfSigma + \epsilon ((\bfz-\bfmu)(\bfz-\bfmu)^T-\bfSigma) + \epsilon^2 (\bfz-\bfmu)(\bfz-\bfmu)^T
\end{equation*}

Denoting by $\lambda_\epsilon$ the leading eigenvalue of $\bfSigma_\epsilon$ and by $\bfu_\epsilon$ the corresponding eigenvector, it holds that
\begin{equation}\label{perturbed:eigen:eq}
\bfSigma_\epsilon\bfu_\epsilon = \lambda_\epsilon \bfu_\epsilon \ \ \text{and} \ \  \bfu_\epsilon^{T}\bfu_\epsilon = 1 \ .
\end{equation}
Next, we expand the perturbed eigenvector and eigenvalue around the unperturbed ones as follows:
\begin{equation*}
\bfu_\epsilon = \bfu_0 + \epsilon \bfu_1 + O(\epsilon^2)
\end{equation*}
and
\begin{equation*}
\lambda_\epsilon = \lambda_0 + \epsilon \lambda_1 + O(\epsilon^2)
\end{equation*}
with
\begin{equation*}
\bfSigma\bfu_0 = \lambda_0  \ \ \text{and} \ \  \bfu_0^{T}\bfu_0 = 1 \ .
\end{equation*}

Substituting the formulas into (\ref{perturbed:eigen:eq}), and equating the zero-th and first order we get
\begin{equation*}
\bfSigma\bfu_0 = \lambda_0 \bfu_0 \ \ \text{and} \ \  \bfu_0^{T}\bfu_0 = 1 \ .
\end{equation*}
and
\begin{equation}\label{1st:order:eq}
((\bfz-\bfmu)(\bfz-\bfmu)^T-\bfSigma)\bfu_0 + \bfSigma\bfu_1 = \lambda_0\bfu_1 + \lambda_1\bfu_0
\end{equation}
and
\begin{equation*}
\bfu_0^{T}\bfu_1 = 0 \ .
\end{equation*}

Multiplying (\ref{1st:order:eq}) from the left with $\bfu_0^T$, we get
\begin{equation*}
\lambda_1 = \bfu_0^T ((\bfz-\bfmu)(\bfz-\bfmu)^T-\bfSigma)\bfu_0 + \bfu_0^T\bfSigma\bfu_1
= (\bfu_0^T (\bfz-\bfmu))^2 - \lambda_0
\end{equation*}

For $\bfu_1$, we rearrange (\ref{1st:order:eq}) to
\begin{equation*}
(\bfSigma-\lambda_0\bfI)\bfu_1 = \lambda_1\bfu_0 - ((\bfz-\bfmu)(\bfz-\bfmu)^T-\bfSigma)\bfu_0
\end{equation*}
and then multiply from the left with the pseudo-inverse of $\bfSigma-\lambda_0\bfI$ to obtain
\begin{equation*}
(\bfSigma-\lambda_0\bfI)^+(\bfSigma-\lambda_0\bfI)\bfu_1 = 
\lambda_1(\bfSigma-\lambda_0\bfI)^+\bfu_0 - (\bfSigma-\lambda_0\bfI)^+((\bfz-\bfmu)(\bfz-\bfmu)^T-\bfSigma)\bfu_0
\end{equation*}
Using the properties (\cite{Mardia&Kent&Bibby:1979}):  $(\bfSigma-\lambda_0\bfI)^+(\bfSigma-\lambda_0\bfI) = \bfI - \bfu_0\bfu_0^T$ and $(\bfSigma-\lambda_0\bfI)^+\bfu_0 = \bfzero$, we obtain
\begin{equation*}
\begin{aligned}
&\bfu_1 - \bfu_0\bfu_0^T\bfu_1 = (\bfSigma-\lambda_0\bfI)^+(\bfz-\bfmu)(\bfz-\bfmu)^T\bfu_0 - (\bfSigma-\lambda_0\bfI)^+ \lambda_0 \bfu_0 \\
\Rightarrow & \bfu_1 = ((\bfz-\bfmu)^T\bfu_0)(\bfSigma-\lambda_0\bfI)^+(\bfz-\bfmu)
\end{aligned}
\end{equation*}
and the proof is completed.
\end{proof}

\subsection{Proof of Proposition \ref{influence:func:cauchy:pca}}
\label{robust:cauchy:proof}
Let us first make the symbolism more explicit and denote $l_F(\bfu|\bftheta)$ the Cauchy log-likelihood function with respect to the distribution $F$ and $\hat{\bfu}_F$ the respective leading Cauchy principal direction. Then, our goal is to calculate the limit of
$$
\frac{1}{\epsilon} (\hat{\bfu}_{F_{\epsilon,\bfz}} - \hat{\bfu}_F)
$$
as $\epsilon\to 0$ where $\hat{\bfu}_{F_{\epsilon,\bfz}}$ is the leading Cauchy principal direction for the distribution $F_{\epsilon,\bfz}=(1-\epsilon)F+\epsilon \Delta_\bfz$. The optimality condition for the leading Cauchy principal direction reads
\begin{equation}
\bfP_{\hat{\bfu}_{F_{\epsilon,\bfz}}} \left. \frac{\partial}{\partial\bfu} l_{F_{\epsilon,\bfz}}\big(\bfu|\bftheta_{F_{\epsilon,\bfz}}(\bfu)\big) \right|_{\bfu=\hat{\bfu}_{F_{\epsilon,\bfz}}} = 0
\label{opt:cond:perturbed}
\end{equation}
and
$$
\bfP_{\hat{\bfu}_{F}} \left. \frac{\partial}{\partial\bfu} l_{F}\big(\bfu|\bftheta_{F}(\bfu)\big) \right|_{\bfu=\hat{\bfu}_{F}} = 0
$$
Moreover, $\hat{\bfu}_{F_{\epsilon,\bfz}}$ is a unit vector which can be represented as
$$
\hat{\bfu}_{F_{\epsilon,\bfz}} = \cos(\rho) \hat{\bfu}_F + \sin(\rho) \bfh
$$
where $\bfh$ is a unit vector perpendicular to $\hat{\bfu}_F$ and $\rho$ is a (small) real number. Under these assumptions, $\hat{\bfu}_{F_{\epsilon,\bfz}}$ is a unit vector since
$$
||\hat{\bfu}_{F_{\epsilon,\bfz}}||_2^2 = \cos^2(\rho) ||\hat{\bfu}_F||_2^2 + \sin^2(\rho) ||\bfh||_2^2 = 1
$$
Obviously, $\rho$ depends on $\epsilon$ and $\bfz$ (i.e., $\rho=\rho(\epsilon,\bfz)$) and $\lim_{\epsilon\to 0} \rho = 0$ but we choose to avoid denoting their explicit relationship because it is not required in our proof. Moreover, a Taylor expansion for the representation leads to
$$
\hat{\bfu}_{F_{\epsilon,\bfz}} = \hat{\bfu}_F + \rho \bfh + O(\rho^2)
$$
thus we obtain that
$$
\bfP_{\hat{\bfu}_{F_{\epsilon,\bfz}}} = \bfP_{\hat{\bfu}_{F}} - \rho (\hat{\bfu}_{F} \bfh^T + \bfh \hat{\bfu}_{F}^T) + O(\rho^2)
$$

Next, we compute the partial derivative using the chain rule
$$
\frac{\partial}{\partial\bfu} l_{F}\big(\bfu|\bftheta_{F}(\bfu)\big) =
\int_{\mathbb R^p} \left[ \frac{\partial}{\partial c} g(c(\bfu), \bftheta_{F}(\bfu)) \frac{\partial}{\partial \bfu} c(\bfu) 
+ \frac{\partial}{\partial \bftheta} g(c(\bfu), \bftheta_{F}(\bfu)) \frac{\partial}{\partial \bfu} \bftheta_{F}(\bfu) \right] dF(\bfx)
$$
Therefore,
\begin{equation*}
\begin{aligned}
\left. \frac{\partial}{\partial\bfu} l_{F}\big(\bfu|\bftheta_{F}(\bfu)\big) \right|_{\bfu=\hat{\bfu}_{F}} &= 
\int_{\mathbb R^p} \left[ \bar{g}_c(\bfx;\hat{\bfu}_{F}) \bfx 
+ \bar{g}_{\bftheta}(\bfx;\hat{\bfu}_{F})
\frac{\partial}{\partial \bfu} \bftheta_{F}(\bfu) \Big|_{\bfu=\hat{\bfu}_{F}} \right] dF(\bfx) \\
&= \int_{\mathbb R^p} \bar{g}_c(\bfx;\hat{\bfu}_{F}) \bfx  dF(\bfx)
+ \int_{\mathbb R^p} \bar{g}_{\bftheta}(\bfx;\hat{\bfu}_{F}) dF(\bfx)
\frac{\partial}{\partial \bfu} \bftheta_{F}(\bfu) \Big|_{\bfu=\hat{\bfu}_{F}} \\
&= \int_{\mathbb R^p} \bar{g}_c(\bfx;\hat{\bfu}_{F}) \bfx  dF(\bfx) 
\end{aligned}
\end{equation*}
The second summand equals to zero because $\hat{\bfu}_{F}$ maximizes the Cauchy log-likelihood function thus it holds that $\int_{\mathbb R^p} \bar{g}_{\bftheta}(\bfx;\hat{\bfu}_{F}) dF(\bfx) = \bfzero$. Similarly,
\begin{equation*}
\begin{aligned}
&\left. \frac{\partial}{\partial\bfu} l_{F_{\epsilon,\bfz}}\big(\bfu|\bftheta_{F_{\epsilon,\bfz}}(\bfu)\big) \right|_{\bfu=\hat{\bfu}_{F_{\epsilon,\bfz}}}
= \int_{\mathbb R^p} \bar{g}_c(\bfx;\hat{\bfu}_{F_{\epsilon,\bfz}}) \bfx  dF_{\epsilon,\bfz} (\bfx) \\
=& (1-\epsilon) \int_{\mathbb R^p} \bar{g}_c(\bfx;\hat{\bfu}_{F_{\epsilon,\bfz}}) \bfx dF(\bfx)
+ \epsilon \bar{g}_c(\bfz;\hat{\bfu}_{F_{\epsilon,\bfz}}) \bfz 
\end{aligned}
\end{equation*}

Next, we further Taylor expand $\bar{g}_{c}(\bfx; \bfu_{F_{\epsilon,\bfz}})$ using $\hat{\bfu}_{F_{\epsilon,\bfz}} = \hat{\bfu}_F + \rho \bfh + O(\rho^2)$
$$
\bar{g}_{c}(\bfx; \hat{\bfu}_{F_{\epsilon,\bfz}}) = \bar{g}_{c}(\bfx; \hat{\bfu}_{F})
+ \rho \bfh \frac{\partial}{\partial\bfu} \bar{g}_{c}(\bfx; {\bfu}) \Big|_{\bfu=\hat{\bfu}_{F}}  + O(\rho^2)
$$
Using again the chain rule, we obtain that
$$
\frac{\partial}{\partial\bfu} \bar{g}_{c}(\bfx; {\bfu}) =
\bar{g}_{cc}(\bfx; {\bfu}) \bfx + \bar{g}_{c\bftheta}(\bfx; {\bfu}) \frac{\partial}{\partial\bfu} \bftheta_{F}(\bfu)
$$
The computation of the partial derivative $\frac{\partial}{\partial \bfu} \bftheta_{F}(\bfu)$ follows. Formula
$\bftheta_{F}(\bfu) = \argmax_{\bftheta} l_F(\bfx^T\bfu|\bftheta)$ implies that
\begin{equation*}
\frac{\partial}{\partial \bftheta} l_F(c(\bfu)|\bftheta) \Big|_{\bftheta=\bftheta_{F}(\bfu)} = 0 \ .
\end{equation*}
Differentiating with respect to $\bfu$ and using the implicit function theorem, we get
\begin{equation*}
\begin{aligned}
\frac{\partial}{\partial \bfu} \bftheta_{F}(\bfu) &=
- \frac{\partial}{\partial \bfu} \frac{\partial}{\partial \bftheta} l_F(c(\bfu)|\bftheta) \Big|_{\bftheta=\bftheta_{F}(\bfu)} \left[ \frac{\partial^2}{\partial \bftheta^2} l_F(c(\bfu)|\bftheta) \Big|_{\bftheta=\bftheta_{F}(\bfu)} \right]^{-1} \\
&= - \int_{\mathbb R^p}  \bfx \bar{g}_{c\bftheta}(\bfx;\bfu) dF(\bfx)
\left[ \int_{\mathbb R^p} \bar{g}_{\bftheta\bftheta}(\bfx;\bfu)dF(\bfx) \right]^{-1}
\end{aligned}
\end{equation*}

Thus,
\begin{equation*}
\begin{aligned}
&\bar{g}_{c}(\bfx; \hat{\bfu}_{F_{\epsilon,\bfz}}) = \bar{g}_{c}(\bfx; \hat{\bfu}_{F}) \\
+& \rho \bfh \left[\bar{g}_{cc}(\bfx; \hat{\bfu}_F) \bfx + \int_{\mathbb R^p}  \bfx \bar{g}_{c\bftheta}(\bfx;\hat{\bfu}_{F}) dF(\bfx)
\left[ \int_{\mathbb R^p} \bar{g}_{\bftheta\bftheta}(\bfx;\hat{\bfu}_{F})dF(\bfx) \right]^{-1} \bar{g}_{c\bftheta}(\bfx; {\hat{\bfu}_F}) \right] + O(\rho^2)
\end{aligned}
\end{equation*}

Overall, (\ref{opt:cond:perturbed}) becomes
\begin{equation*}
\begin{aligned}
&\left[ \bfP_{\hat{\bfu}_F} - \rho (\hat{\bfu}_{F} \bfh^T + \bfh \hat{\bfu}_{F}^T) + O(\rho^2) \right] \cdot
\left[ (1-\epsilon) \int_{\mathbb R^p} \bar{g}_c(\bfx;\hat{\bfu}_{F_{\epsilon,\bfz}}) \bfx dF(\bfx)
+ \epsilon \bar{g}_c(\bfz;\hat{\bfu}_{F_{\epsilon,\bfz}}) \bfz  \right] = 0 \\
\Rightarrow&
\bfP_{\hat{\bfu}_F} \int_{\mathbb R^p} \bar{g}_c(\bfx;\hat{\bfu}_{F_{\epsilon,\bfz}}) \bfx dF(\bfx)
- \rho (\hat{\bfu}_{F} \bfh^T + \bfh \hat{\bfu}_{F}^T) \int_{\mathbb R^p} \bar{g}_c(\bfx;\hat{\bfu}_{F}) \bfx dF(\bfx) + O(\rho^2) \\
&= \epsilon \bfP_{\hat{\bfu}_F} \left[\int_{\mathbb R^p} \bar{g}_c(\bfx;\hat{\bfu}_{F}) \bfx dF(\bfx)
- \bar{g}_c(\bfz;\hat{\bfu}_{F_{\epsilon,\bfz}}) \bfz \right] + O(\epsilon\rho) \\
\Rightarrow&
\rho \bfh \left[ \int_{\mathbb R^p}\bar{g}_{cc}(\bfx; \hat{\bfu}_F) \bfx^T \bfx dF(\bfx)
+ \int_{\mathbb R^p} \bar{g}_c(\bfx;\hat{\bfu}_{F}) \hat{\bfu}_{F}^T \bfx dF(\bfx) \right. \\
&+ \left. \int_{\mathbb R^p}  \bfx \bar{g}_{c\bftheta}(\bfx;\hat{\bfu}_{F}) dF(\bfx)
\left[ \int_{\mathbb R^p} \bar{g}_{\bftheta\bftheta}(\bfx;\hat{\bfu}_{F})dF(\bfx) \right]^{-1} \int_{\mathbb R^p} \bar{g}_{c\bftheta}(\bfx; {\hat{\bfu}_F}) \bfx dF(\bfx) \right] + O(\rho^2) \\
&= \epsilon \bfP_{\hat{\bfu}_F} \left[\int_{\mathbb R^p} \bar{g}_c(\bfx;\hat{\bfu}_{F}) \bfx dF(\bfx)
- \bar{g}_c(\bfz;\hat{\bfu}_{F_{\epsilon,\bfz}}) \bfz \right] + O(\epsilon\rho)
\end{aligned}
\end{equation*}
where we use the facts that
$$\bfP_{\hat{\bfu}_F} \bfh = \bfh$$
and 
$$\bfh^T \int_{\mathbb R^p} \bar{g}_c(\bfx;\hat{\bfu}_{F}) \bfx dF(\bfx)
= \bfP_{\hat{\bfu}_F} \int_{\mathbb R^p} \bar{g}_c(\bfx;\hat{\bfu}_{F}) \bfx dF(\bfx)
= \bfP_{\hat{\bfu}_F} \left. \frac{\partial}{\partial\bfu} l_{F}\big(\bfu|\bftheta_{F}(\bfu)\big) \right|_{\bfu=\hat{\bfu}_{F}} = 0
$$

Thus, the influence function is
$$
IF_{\hat{\bfu}_F} (\bfz, F) = \lim_{\epsilon\to0} \frac{\rho\bfh}{\epsilon} = \bfA^{-1} \bfb
$$
where
\begin{equation*}
\begin{aligned}
\bfA &= \bfI_{d} \left[ \int_{\mathbb R^p}\bar{g}_{cc}(\bfx; \hat{\bfu}_F) \bfx^T \bfx dF(\bfx)
+ \int_{\mathbb R^p} \bar{g}_c(\bfx;\hat{\bfu}_{F}) \hat{\bfu}_{F}^T \bfx dF(\bfx) \right] \\
&+ \int_{\mathbb R^p}  \bfx \bar{g}_{c\bftheta}(\bfx;\hat{\bfu}_{F}) dF(\bfx)
\left[ \int_{\mathbb R^p} \bar{g}_{\bftheta\bftheta}(\bfx;\hat{\bfu}_{F})dF(\bfx) \right]^{-1} \int_{\mathbb R^p} \bar{g}_{c\bftheta}(\bfx; {\hat{\bfu}_F}) \bfx dF(\bfx)
\end{aligned}
\end{equation*}
and
$$
\bfb = \bfP_{\hat{\bfu}_F} \left[\int_{\mathbb R^p} \bar{g}_c(\bfx;\hat{\bfu}_{F}) \bfx dF(\bfx)
- \bar{g}_c(\bfz;\hat{\bfu}_{F_{\epsilon,\bfz}}) \bfz \right]
$$

\clearpage

\end{document}